\documentclass[a4paper]{article}

\usepackage{amsthm, amsfonts, amsmath}
\usepackage{graphicx}
\usepackage{appendix}
\usepackage{algorithm}
\usepackage{algorithmicx}
\usepackage[noend]{algpseudocode}
\usepackage{booktabs}
\usepackage{enumerate}
\usepackage{mathtools}
\usepackage{xspace} 
\usepackage{enumitem}
\usepackage[hidelinks]{hyperref}
\usepackage[compress]{cite} 

\RequirePackage[T1]{fontenc} \RequirePackage[tt=false, type1=true]{libertine} \RequirePackage[varqu]{zi4} \RequirePackage[libertine]{newtxmath}

\newtheorem{lemma}{Lemma}
\newtheorem{theorem}[lemma]{Theorem}

\newtheorem{definition}[lemma]{Definition}

\newcommand{\bigO}{\smash{\ensuremath{O}}}
\newcommand{\tilO}{\smash{\ensuremath{\widetilde{O}}}}
\newcommand{\tilOm}{\smash{\ensuremath{\widetilde{\Omega}}}}
\newcommand{\tilT}{\smash{\ensuremath{\widetilde{\Theta}}}}
\newcommand{\tild}{\smash{\ensuremath{\widetilde{d}}}}

\newcommand{\hybrid}{\ensuremath{\mathsf{HYBRID}}\xspace}
\newcommand{\HYBRID}{\ensuremath{\mathsf{HYBRID}}\xspace}

\newcommand{\hybridpar}[2]{\ensuremath{\mathsf{HYBRID}(#1,#2)}}
\newcommand{\hybridparbig}[2]{\ensuremath{\mathsf{HYBRID}\big(#1,#2\big)}}
\newcommand{\LOCAL}{\ensuremath{\mathsf{LOCAL}}\xspace}
\newcommand{\CONGEST}{\ensuremath{\mathsf{CONGEST}}\xspace}

\newcommand{\NCC}{\ensuremath{\mathsf{NCC}}\xspace}

\newcommand{\CC}{\ensuremath{\mathsf{Congested\, Clique}}\xspace}

\newcommand{\MPC}{\ensuremath{\mathsf{MPC}}\xspace}

\newcommand{\eps}{\varepsilon}
\newcommand{\calA}{\mathcal{A}}

\newcommand{\calS}{\mathcal{S}}

\newcommand{\p}{\ensuremath{\!+\!}}

\DeclareMathOperator{\polylog}{polylog}
\DeclareMathOperator{\poly}{poly}
\DeclareMathOperator{\hop}{hop}
\DeclareMathOperator*{\argmin}{arg\,min}

\begin{document}
	
	\title{Near Tight Shortest Paths in the Hybrid Model}
	\date{}
	
	\author{Philipp Schneider, University of Bern, Switzerland\\ (philipp.schneider2@unibe.ch)}
	
	\maketitle
	
	\begin{abstract}
		Shortest paths problems are subject to extensive studies in classic distributed models such as the \CONGEST or \CC. These models dictate how nodes may communicate in order to determine shortest paths in a distributed input graph.
		This article focuses on shortest paths problems in the \hybrid model, which combines local communication along edges of the input graph with global communication between arbitrary pairs of nodes that is restricted in terms of bandwidth.
		
		Previous work 
		showed that it takes $\tilOm(\!\sqrt{k})$ rounds in the \hybrid model for each node to learn its distance to $k$ dedicated source nodes (aka the $k$-SSP problem), even for crude approximations. 
		This lower bound was also matched with algorithmic solutions for $k \geq n^{2/3}$. 
		However, as $k$ gets smaller, the gap between the known upper and lower bounds diverges and even becomes exponential for the single source shortest paths problem (SSSP).
		In this work we close this gap for the whole range of $k$ (up to terms that are polylogarithmic in $n$), by giving algorithmic solutions for $k$-SSP in $\tilO\big(\!\sqrt k\big)$ rounds for any $k$.
	\end{abstract}
	
	\section{Introduction}
	
	Computing or approximating shortest paths in a network has important uses, for instance to obtain information about the topology of the network or as a subroutine for related tasks like computing or updating tables for IP routing.
	In shortest paths problems nodes of a network are required to learn their distance to other nodes and in the distributed version of the shortest paths problem knowledge of input graph is spread over the nodes, which must then determine their distance to other nodes by efficiently communicating with each other. 
	
	This problem has obtained significant attention from a theoretical angle in the classical models of distributed computing over the last decade (see e.g.\ \cite{Nanongkai2014, Dory2021} and the references therein). 
	These classic models come in two different ``flavors''. The first is based on  \textit{local communication} in a graph, where nodes communicate in synchronous rounds and in each such round adjacent pairs of nodes in a graph $G$, which also serves as the problem input, may exchange a small message (this is called \CONGEST model).
	
	The second flavor captures \textit{global communication}, where any pair of nodes can exchange a small message per round and each node initially knows only its neighbors of the input graph $G$ (this became known as the \CC model).  Another example of global communication is the \MPC model (cf.\ \cite{Karloff2010, Goodrich2011,Andoni2014,Beame2017}), which in some sense generalizes the \CC model and was inspired by practical applications (cf. MapReduce \cite{Dean2008}). 
	
	
	Theoretic research on distributed graph problems and in particular shortest paths problems has concentrated mostly on distributed models where nodes use only either a local or global mode of communication. This began to change only recently with the introduction of the so called \hybrid model \cite{Augustine2020a}, where nodes have access to two modes of communication. First, a graph based local mode which allows neighbors to communicate by exchanging relatively large messages. And second, a global mode where any pair of nodes can communicate in principle but only a few small messages may be exchanged each round.
	
	The hybrid model is intended to capture the fact that in modern networks, nodes can often interface with multiple, diverse communication infrastructures. For instance, data centers combine wired communication with wireless communication \cite{Halperin2011}. Organizations can augment their own networks with communication over the Internet using virtual private networks (VPNs) \cite{Rossberg2011}. Another example are mobile devices like smart-phones, which have access to high bandwidth, short range communication via WiFi or Bluetooth with other devices that are close by, which can be combined with the relatively low-bandwidth communication via the cellular network (cf., 5G \cite{Asadi2016}). 
	
	\subsection{The \hybrid model}
	
	For the formal definition of the \hybrid model, we rely on  the concept of \textit{synchronous message passing} \cite{Lynch1996}, where nodes exchange messages and conduct local computations in synchronous rounds. Note that synchronous message passing focuses on the \textit{round complexity}, i.e., the number of rounds required to solve a distributed problem, and therefore nodes are considered computationally unbounded.
	
	\begin{definition}[Synchronous Message Passing, cf.\ \cite{Lynch1996}] 
		\label{def:sync_msg_passing}
		Let $V$ be a set of $n$ nodes with unique identifiers ID\smash{$:V \!\to\! [n] \!\stackrel{\text{def}}{=}\! \{1, \ldots , n\}$}. Time is slotted into discrete rounds. Nodes wake up synchronously at the start of some round and each round consists of the following steps. First, all nodes receive the set of messages addressed to them in the last round. Second, nodes conduct computations based on their current state and the set of received messages to compute their new state (randomized algorithms also include the result of a random function). Third, based on the new state, the next messages are sent.
	\end{definition}
	
	The aim of the  \hybrid model is to reflect the fundamental concepts of \textit{locality} and \textit{congestion} to capture the nature of distributed systems that combine {both} {physical} and {logical} networks.
	%
	%
	
	\begin{definition}[cf. \cite{Augustine2019}]
		\label{def:hybrid_model}
		The \hybridpar{\lambda}{\gamma} model is a synchronous message passing model (Def.\ \ref{def:sync_msg_passing}), subject to the following restrictions. \emph{Local mode:} nodes may send a message per round of maximum size $\lambda$ bits to each of their neighbors in a connected graph. \emph{Global mode:} nodes can send and receive messages of total size at most $\gamma$ bits per round to/from any other node(s) in the network. If these restrictions are violated a strong adversary\footnote{The strong adversary knows the states of all nodes, their source codes and even the outcome of all random functions.} selects the messages that are delivered.
	\end{definition}
	
	The parameter $\lambda$ restricts the bandwidth over {edges} in the local network, and $\gamma$ restricts the amount of global communication of {nodes}. Notably, the classical models of distributed computing are covered by this model as marginal cases: \LOCAL and \CONGEST are given by $\gamma = 0$ and $\lambda = \infty$ and $\lambda \in \bigO(\log n)$, respectively. The \CC and \NCC models are given by $\lambda = 0$ and $\gamma \in \bigO(n \log n)$ (due Lenzens  routing algorithm \cite{Lenzen2013})  and $\gamma \in \bigO(\log^2 n)$, respectively.
	
	Of particular practical and theoretical interest are non-marginal parameterizations of \hybridpar{\lambda}{\gamma} that push both communication modes to one extreme end of the spectrum. More specifically, to model the high bandwidth of physical connections we leave the size of local messages unrestricted. To model the severely restricted global bandwidth of shared logical networks, we allow only $\polylog n$ bits of global communication per node per round. Formally, we define the ``standard'' hybrid model as combination of the standard \LOCAL \cite{Peleg2000} and node capacitated clique \cite{Augustine2019} models: $\hybrid := \hybridparbig{\infty}{\bigO(\log^2 n)}$.
	
	\subsection{Preliminaries}
	\label{sec:preliminaries}
	
	We continue with some definitions, conventions and nomenclature that we will use in the following.
	
	\begin{definition}[$k$-{Sources Shortest Paths} ($k$-SSP) Problem]
		\label{def:kSSP}
		Given subset of $k$ source nodes in a graph $G=(V,E)$, every node $v \in V$ has to learn $d(v,s)$ for all sources $s$. In the $\alpha$-approximate version of the problem for \emph{stretch} $\alpha \geq 1$, every $u\in V$ has to learn values $\tilde{d}(v,s)$ such that $d(v,s)\leq \tilde{d}(v,s)\leq \alpha d(v,s)$ for all sources $s$.
	\end{definition}
	
	The \textit{all-pairs shortest paths problem} (APSP) equals the case $k=n$ and the single-source shortest paths problem (SSSP) equals the case $k=1$. Note that the local communication graph and the input graph for the graph problem are the same, which is a standard assumption for distributed models with graph-based communication (like \LOCAL and \CONGEST).
	
	Our algorithms are randomized and aim for success \textit{with high probability} (w.h.p.), which means with success probability at least $1-\frac{1}{n^c}$ for an arbitrary constant $c>0$. We write i.i.d.\ if we pick elements from some set \textit{independently, identically distributed}. In this work, logarithm functions without subscript are generally to the base of two, i.e., \smash{$\log \stackrel{\text{def}}{=} \log_2$}. Sometimes we write $\polylog n$ to describe terms of the form $q(\log n)$ where $q$ is a polynomial. We abbreviate sets of the form $\{1, \dots, k\},  k \in \mathbb{N}$ with $[k]$. We will often neglect logarithmic factors in $n$ using the soft $\tilO$-notation.
	
	We consider undirected, connected communication graphs $G = (V,E)$. Edges have {weights} $w: E \to [W]$, where $W$ is at most polynomial in $n$, thus the weight of an edge and of a simple path fits into a $\bigO(\log n)$ bit message. A graph is considered {unweighted} if $W=1$. Let $w(P) = \sum_{e \in P}w(e)$ denote the length of a path $P \subseteq E$. 
	
	Then the \emph{distance} between two nodes $u,v \in V$ is
	{$
		d_G(u,v) := \!\min_{\text{$u$-$v$-path } P} w(P).
		$}
	A path with smallest length between two nodes is called a \emph{shortest path}.
	Let $|P|$ be the number of edges (or \emph{hops}) of a path $P$.
	We define the \emph{$h$-hop limited distance} from $u$ to $v$ as
	$
	d_{G,h}(u,v) := \!\!\min_{{\text{$u$-$v$-path } P, |P| \leq h }}\, w(P).
	$
	If there is no $u$-$v$ path $P$ with $|P|\leq h$ we define $d_{G,h}(u,v) := \infty$.
	
	The \emph{hop-distance} between two nodes $u$ and $v$ is defined as
	{$
		\hop_G(u,v) := \!\min_{\text{$u$-$v$-path } P} |P|.
		$ }
	We generalize this for sets $U,W \subseteq V$
	{$
		\hop_G(U,W) := \!\min_{u \in U, w \in W} \hop_G(u,w)
		$}  (whereas $\hop_G(v,v) := 0$).
	The \emph{diameter} of $G$ is defined as
	{$
		D_G:= \max_{u,v \in V} \hop_G(u,v).
		$}
	We drop the subscript $G$, if $G$ is clear from the context.

	\subsection{Related Work}
	
	\subparagraph{Local Communication Networks.} Shortest paths problems have been intensely studied in the \CONGEST model. For the SSSP problem, \cite{Sarma2012} showed that any approximation algorithm has a runtime of $\tilOm(\!\sqrt{n} + D)$ for any constant stretch.
	Following the publication of this lower bound, there has been a series of papers that attempt to obtain algorithms that get close to the lower bound, see \cite{Lenzen2013a,Nanongkai2014,Henzinger2019,Becker2017,Elkin2017,Ghaffari2018,Forster2018,Chechik2021,CensorHillel2021a}. Notably, the work of \cite{Becker2017} gives an algorithm that computes a $(1+\varepsilon)$-approximate SSSP solution in time $\tilO(\!\sqrt{n}\!+\!D)$.
	
	In another strain of work, \cite{Haeupler2018, Zuzic2022} aim for SSSP algorithms that are competitive in terms of complexity to the best algorithm that can be crafted for a given graph topology, which is also known as \textit{universal optimality}. The work of \cite{Zuzic2022} achieves a $(1+\varepsilon)$ approximation, in time $T\!\cdot\! n^{o(1)}$, where $T$ is the complexity of the best \CONGEST algorithm for the given topology, which comes close to the best known algorithm that is optimized for the worst case where $T \in \tilOm(\!\sqrt{n} + D)$ (see above).
	
	

	The APSP problem in \CONGEST has received significant attention as well \cite{Holzer2012, Frischknecht2012, Nanongkai2014, Lenzen2015, Abboud2016,Hua2016,Elkin2017,CensorHillel2017,Huang2017,Agarwal2018,Agarwal2019,Bernstein2019,Agarwal2020,Abboud2021}.
	Linear lower bounds that hold for small diameter graphs are shown or implied by \cite{Frischknecht2012, Nanongkai2014, Abboud2016, CensorHillel2017}. Notably a lower bound of $\tilOm(n / \log n)$ \cite{Nanongkai2014} that holds for the unweighted case and constant stretch approximations is matched by an $\bigO(n / \log n)$ algorithm for unweighted graphs \cite{Hua2016}. Interestingly, \cite{Abboud2021} gives a $\Omega(n)$ lower bounded for the weighted problem, thus separating it from the unweighted case.
	In terms of algorithmic upper bounds, the current state of the art is the randomized algorithm by \cite{Bernstein2019}, which computes an exact solution in $\tilO(n)$ rounds, even for directed edges with negative and zero weights. The best deterministic, exact algorithm takes $\tilO(n^{4/3})$ rounds and is due to \cite{Agarwal2020}.

	\subparagraph{Global Communication Networks.}
	Shortest paths problems also obtained significant attention in the \CC model, see,
	\cite{LeGall2016, CensorHillel2018, CensorHillel2019, CensorHillel2019a,Izumi2019,Dory2020,Biswas2021,Dory2021}. Some of these results give fast (down to $\tilO(1)$ rounds), constant approximations for APSP and $k$-SSP in the \CC model. The power of the \CC model is demonstrated by \cite{Dory2020} which achieves a $\poly(\log \log n)$ round, constant approximation of APSP on \textit{unweighted} graphs. For cruder $\tilO(1)$ approximations of APSP \cite{Dory2021} even achieves constant rounds on weighted graphs.
	
	In a branch similar to \CC, the work of \cite{Nanongkai2014,Holzer2016,Becker2017} investigates shortest paths in the \textit{broadcast} variant of \CC, where in each round, each node may broadcast a single message. Notably, \cite{Becker2017} gives a $(1\!+\!\varepsilon)$-approximation algorithm for the SSSP problem that runs in $\tilO(1)$ rounds.
	Finally, \cite{Augustine2019} introduces the more restrictive \NCC model and considers various fundamental graph problems. Their algorithmic solutions include the computation of BFS trees, which can be used to solve the SSSP problem for \textit{unweighted} graphs of \textit{bounded arboricity} in $\tilO(D)$ rounds. 
	
	\begin{figure}[h]
		\centering
		\includegraphics[width=0.87\textwidth]{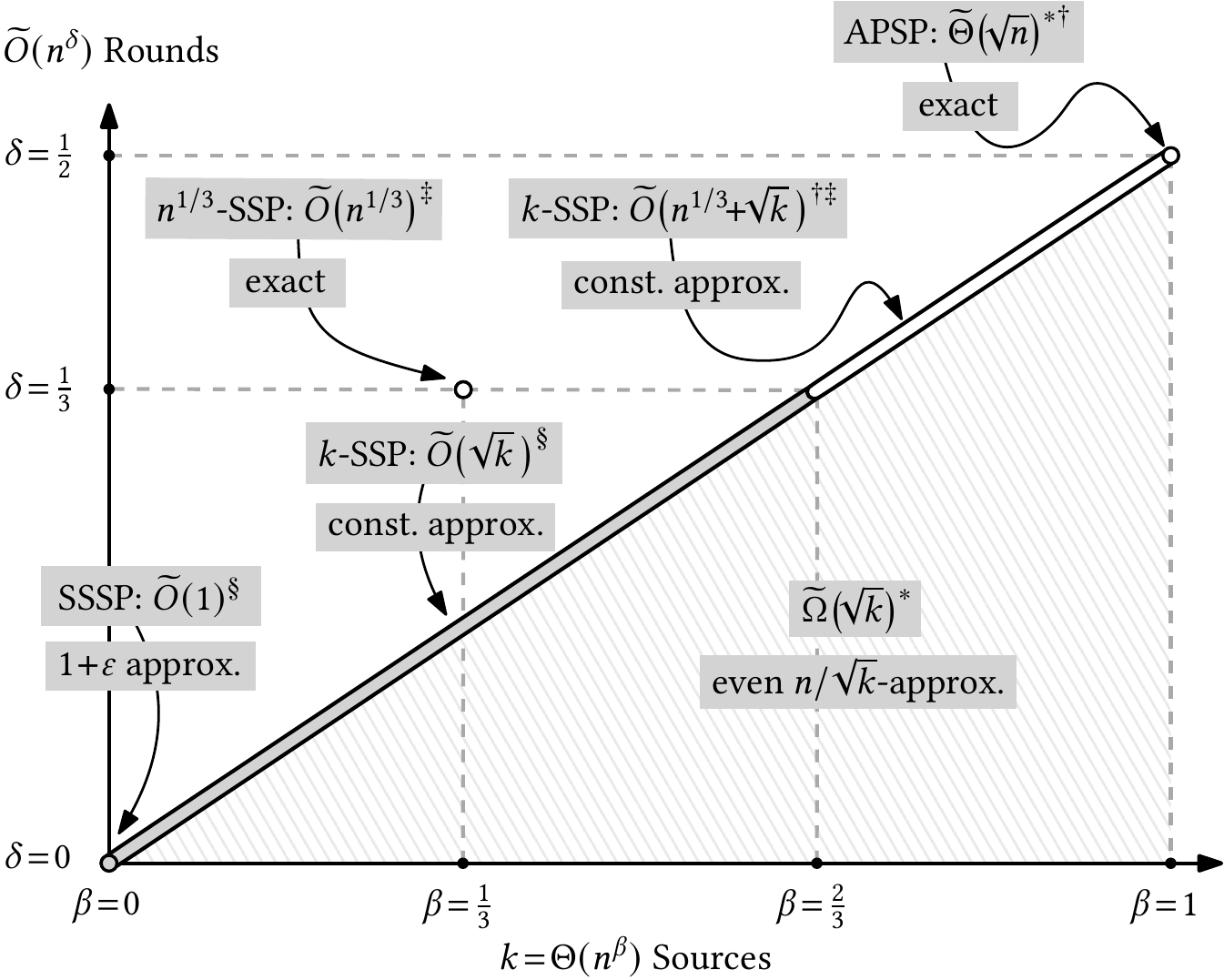}
		\caption{Complexity landscape of the $k$-SSP problem with the number of sources on the horizontal and the round complexity on the vertical axis. Circles or bars denote known upper bounds (ours in gray). The gray shaded area denotes the lower bound. References are as follows $*$: \cite{Augustine2020a}, $\dagger$: \cite{Kuhn2020}, $\ddagger$: \cite{CensorHillel2021}, \S: this work.}
		\label{fig:kssp_complexity_overview}
	\end{figure}
	
	\subparagraph{Hybrid Communication Networks}
	A graphic overview of the state of the art for the $k$-SSP problem is provided in Figure \ref{fig:kssp_complexity_overview}. Research on shortest paths in the \hybrid model was started by \cite{Augustine2020a}, with a range of results. They introduce a technique for efficiently broadcasting messages in the network, which they leverage, for instance, to obtain an exact solution of APSP in $\tilO(n^{2/3})$ rounds and a $1\p \eps$ approximation in $\tilO(\!\sqrt n)$ rounds (for constant $\eps >0$), matching their lower bound of $\tilOm(\!\sqrt n)$ (which holds even for randomized $\bigO(\!\sqrt n)$ approximations). 
	For SSSP, they give a $1 \p \eps$ approximation in $\tilO(n^{1/3})$ rounds, a $(1/\eps)^{\bigO(1/\eps)}$ approximation in $\tilO(n^\eps)$ rounds and an exact solution in $\tilO\big(\!\sqrt{SPD}\big)$ rounds, where $SPD$ is the shortest path diameter. 
	
	Subsequently, \cite{Kuhn2020} improves APSP to $\tilO(n^{1/2})$ for an exact solution based on a scheme for efficiently unicasting messages in the network, which settled the problem (up to logarithmic factors in the round complexity). They also give an exact solution for SSSP in \smash{$\tilO(n^{3/5})$} rounds and constant stretch  approximations for $k$-SSP in \smash{$\tilO(n^{1/3}+\sqrt{k})$} rounds, which is tight with their lower bound of \smash{$\tilOm(\!\sqrt k)$} rounds for $k \geq n^{2/3}$. 
	The stretch was later improved by \cite{CensorHillel2021} which combines the techniques of \cite{Kuhn2020} with a density sensitive approach. In particular, \cite{CensorHillel2021} gives exact solutions for $k \geq n^{2/3}$ \textit{random} sources in \smash{$\tilO(\!\sqrt k)$} rounds and for $n^{1/3}$-SSP in $\tilO(n^{1/3})$ rounds.
	
	In a subsequent work, \cite{CensorHillel2021a} uses density awareness in a different way to solve SSSP in \smash{$\tilO(n^{5/17})$} rounds for stretch $(1\!+\!\eps)$. 
	A deterministic protocol for efficiently broadcasting local edges is given by \cite{Anagnostides2021} and then used to obtain a deterministic APSP-algorithm with stretch \smash{$\frac{\log n}{\log\log n}$} in \smash{$\tilO(n^{1/2})$} rounds.
	For classes of sparse graphs (e.g., cactus graphs), \cite{Feldmann2020} demonstrates that $\polylog n$ solutions are possible even with \CONGEST as local network.
	
	\subsection{Contributions}

	Shortest paths problems are hard to solve if only either of the constituent communication modes of the \hybrid model can be used in isolation. In the \LOCAL model computing even crude approxiations for a single source requires $D_G$ rounds, where $D_G$ is the diameter of $G$ which can be linear in $n$. In the \NCC model, a single node learning its distance to $k$ source nodes has a lower bound of $\tilOm(k)$. 
	By contrast, the known bounds for the $k$-SSP problem (see Definition \ref{def:kSSP}) in the \hybrid model are \smash{$\tilOm(\!\sqrt{k})$} for $k \in [n]$ and \smash{$\tilO(\!\sqrt{k})$} for relatively large $k$ \cite{Augustine2020a, Kuhn2020}. This work focuses on closing the gap to the lower bound for the whole range of $k \in [n]$ (up to $\tilO(1)$ factors). In particular, we show the following.
	
	\begin{itemize}
		\item There exists an algorithm that solves the single sources shortest paths problem with stretch $(1 \!+\! \eps)$ w.h.p.\ in $\tilO(1)$ rounds in the \hybrid model (Theorem \ref{thm:almost_shortest_sssp}).
		\item There exists an algorithm that solves the $k$ \textit{random} sources shortest paths problem with stretch $(1 \!+\! \eps)$ w.h.p.\ in \smash{$\tilO\big(\!\sqrt{k}\big)$} rounds in the \hybrid model (Thm.\!~\ref{thm:k-ssp}).
		\item There exists an algorithm that solves the $k$ sources shortest paths problem with stretch $(3 \!+\! \eps)$ w.h.p.\ in \smash{$\tilO\big(\!\sqrt{k}\big)$} rounds in the \hybrid model (Theorem \ref{thm:k-ssp}).
	\end{itemize}
	
	The currently fastest algorithms for constant stretch SSSP are by \cite{Augustine2020a} which takes  $\bigO(n^\eps)$ rounds (where $\eps$ is required to be constant to keep a constant stretch) and by \cite{CensorHillel2021a} for stretch $1 \p \eps$ in \smash{$\tilO(n^{5/17})$} rounds. Both solutions are exponentially slower than the $1 \p \eps$ approximation proposed in this work. The known lower bound for $k$-SSP of \smash{$\tilOm\big(\!\sqrt{k}\big)$}, which holds even for randomized algorithms and extremely large stretch in \smash{$\Omega\big(\!\sqrt{n}\big)$} was matched (up to $\tilO(1)$ factors) for $k \geq n^{2/3}$ by a series of articles \cite{Augustine2020a, Kuhn2020, CensorHillel2021}, with varying, constant stretches.
	However, tight algorithms below this threshold for $k$ have been elusive and this work answers the question whether the lower bound of \smash{$\tilOm\big(\!\sqrt{k}\big)$} can be matched for $k < n^{2/3}$ positively. The current complexity landscape of the $k$-SSP problem in \hybrid is given by Figure \ref{fig:kssp_complexity_overview}. Furthermore, we generalize our solutions for $k$-SSP for the \hybridpar{\infty}{\gamma} model and show that approximations can be obtained in \smash{$\tilO\big(\!\sqrt{k/\gamma}\big)$} rounds\footnote{Note that this is also tight up to $\tilO(1)$ factors due to a matching lower bound of \smash{$\tilO\big(\!\sqrt{k/\gamma}\big)$} by \cite{Schneider2023} for $k$-SSP, which is a slight generalization of the lower bound of \cite{Kuhn2020}.} (details in Theorem \ref{thm:k-ssp}), demonstrating the benefit of larger global capacity for solving $k$-SSP.

	In terms of technical contributions, we show that an interface model called  $\mathsf{Minor}$-$\mathsf{Aggregation}$ and an oracle to solve the $\mathsf{Eulerian}$-$\mathsf{Orientation}$ problem that were introduced by \cite{Rozhon2022} (for SSSP solutions in the $\mathsf{PRAM}$ model), can be efficiently implemented in the \hybrid model, which might be useful in case other problems have fast solutions in this interface model. Furthermore, we show that multiple graph algorithms can be scheduled efficiently in parallel on an appropriately sized skeleton graph. This is useful for our purposes since we can reduce $k$-SSP to multiple instances of SSSP on such a skeleton graph, but might prove equally useful for other graph problems that can be broken into smaller problems on skeleton graphs in a similar way.

	\section{Almost Optimal SSSP in Polylogarithmic Time}
	\label{sec:sssp_logtime}
	
	In this section we show how  to obtain an almost optimal $\tilO(1)$ round algorithm for the \hybrid model by adapting the techniques of \cite{Rozhon2022}. In particular, we will prove the following theorem.
	
	\begin{theorem}
		\label{thm:almost_shortest_sssp}
		A $(1 \p \eps)$-approximation of SSSP can be computed in $\tilO(1/\eps^2)$ rounds in the \hybrid model, w.h.p.
	\end{theorem}

	The goal of \cite{Rozhon2022} is the computation of an almost shortest path tree in the $\mathsf{PRAM}$ model with linear work and $\tilO(1)$ depth. 
	Their main machinery is the simulation of an interface model, called $\mathsf{Minor}$-$\mathsf{Aggregation}$ model, defined as follows. Let $G=(V,E)$ an undirected input graph. In a difference to the synchronous message passing model (Definition \ref{def:sync_msg_passing}) nodes \textit{and edges} are assumed to be individual computational units that communicate in synchronous rounds and conduct arbitrary local computations in each round. Initially, nodes know their ID of size $\tilO(1)$ and edges know the IDs of their endpoints. Each round, communication occurs by conducting the following operations (in that order).
	
	\begin{itemize}
		\item \textbf{Contraction:} Each edge e chooses a value $c_e \in \{\top, \bot\}$, which defines a minor network $G' = (V', E')$ where all edges with $c_e = \top$ are contracted. This forms a set supernodes \smash{$V'\subseteq 2^{V}$}, where a supernode $s \in V'$ consists of nodes connected by contracted edges. For each $e \in E$ that connects nodes of distinct supernodes, there is a corresponding edge in $E'$. 
		\item \textbf{Consensus:} Each $v \in V$ chooses a $\tilO(1)$-bit value $x_v$. For each supernode $s \in V'$, let \smash{$y_s := \bigoplus_{v \in s} x_v$}, where $\bigoplus$ is some pre-defined aggregation operator. All
		$v \in s$ learn $y_s$.
		\item \textbf{Aggregation:} Each edge $e \in E'$ connecting supernodes $a \in V'$ and $b \in V'$ learns $y_a$
		and $y_b$, and chooses two $\tilO(1)$-bit values $z_{e,a}$, $z_{e,b}$ (i.e., one for each endpoint). Every node $v \in s$ of each supernode $s \in V'$ learns the aggregate of its incident edges in $E'$, i.e.,
		\smash{$\bigotimes_{e \in \text{incidentEdges}(s)} z_{e,s}$} where $\bigotimes$ is some pre-defined aggregation operator. All nodes $v \in s$
		learn the same aggregate value (if the aggregate is not unique).
	\end{itemize}
	
	The work of \cite{Rozhon2022} shows that SSSP can be solved with a stretch of $(1 \p \eps)$ in $\tilO(1/\eps^2)$ rounds of the $\mathsf{Minor}$-$\mathsf{Aggregation}$ model, \textit{if} it can call on an oracle $\mathcal O^{\text{Euler}}$ that solves the $\mathsf{Eulerian}$-$\mathsf{Orientation}$ problem once per round (we define this oracle formally further below). 
	
	\begin{lemma}[cf., \cite{Rozhon2022}] 
		\label{lem:sssp_minor_aggregation}
		A $(1 \p \eps)$-approximation of SSSP on $G$ can be computed with a total of $\tilO\big(1/\eps^2\big)$ rounds of $\mathsf{Minor}$-$\mathsf{Aggregation}$ model and calls to the oracle $\mathcal O^{\text{Euler}}$ on certain Eulerian graphs $H$ (see Definition \ref{def:euler_oracle}), respectively.
	\end{lemma}
	
	Our goal is to apply this lemma to show Theorem \ref{thm:almost_shortest_sssp}. For this we have to prove that we can (1) efficiently simulate the $\mathsf{Minor}$-$\mathsf{Aggregation}$ model in \hybrid and 
	(2) efficiently implement the oracle $\mathcal O^{\text{Euler}}$ in \hybrid.

	\subsection{Simulation of $\mathsf{Minor}$-$\mathsf{Aggregation}$ in \hybrid}
	
	We show in the following that the $\mathsf{Minor}$-$\mathsf{Aggregation}$ model can be implemented efficiently in the \hybrid model. One of the techniques we employ is  the following result by \cite{Goette2021}, which states that an ``overlay network'' can be computed with diameter, degree and round complexity \smash{$\tilO(1)$}.
	
	\begin{lemma}[see \cite{Goette2021}]
		\label{lem:overlay}
		Given that nodes know their neighbors in a graph $G = (V,E)$ as problem input and a polynomial upper bound of $n$, a rooted tree $T=(V,E_T)$ (usually not a sub-graph of $G$) with constant degree and diameter $\bigO(\log n)$ can be computed in $\bigO(\log n)$ rounds, w.h.p., in the \hybrid model (\NCC suffices).
	\end{lemma}
	
	We proceed with the simulation of the $\mathsf{Minor}$-$\mathsf{Aggregation}$ model in \hybrid.
	
	\begin{lemma}
		\label{lem:simulate_minor_agregation}
		A round of the $\mathsf{Minor}$-$\mathsf{Aggregation}$ model can be simulated in $\tilO(1)$ rounds in the \hybrid model, w.h.p.
	\end{lemma}
	
	\begin{proof}
		Since nodes in the distributed model are computationally unlimited, each edge can be simulated by both of its endpoints simultaneously. We compute a tree $T_s$ with diameter, degree and round complexity \smash{$\tilO(1)$} on each connected component $s \in V'$ using the set of contracted edges ($e\in E$ with $c_e = \top$) as input, see Lemma \ref{lem:overlay}. Note that the small diameter and degree of $T_s$ allows us efficient computation of aggregation operations by conducting converge-casts on $T_s$ in $\tilO(1)$ rounds.
		This also allows us to implement the consensus step by converge casting the values $y_s$ of for all $s \in V'$ to all $v \in s$.
		
		As part of the simulation, we assume that each edge $e = \{s_1,s_2\} \in E'$ is simulated by \textit{both} of its {incident} nodes $v_1 \in s_1,v_2 \in s_2$. The values $y_{s_1}, y_{s_2}$ from the consensus step can be exchanged between $v_1,v_2$ with their shared edge in a single round so both can continue the correct simulation of $e$. This is required in order to choose values $z_{e,a}$, $z_{e,b}$, which may depend on $y_{s_1}, y_{s_2}$. The aggregation of values $z_{e,s}$ of all $e \in E'$ incident to some $s\in V'$ is implemented as before.
	\end{proof}
	
	\subsection{Implementation of the Oracle $\mathcal O^{\text{Euler}}$}
	
	The oracle $\mathcal O^{\text{Euler}}$ solves the $\mathsf{Eulerian}$-$\mathsf{Orientation}$ problem on certain Eulerian graphs $H$ (i.e., graphs containing an Eulerian cycle), which requires that all edges are oriented such that in- and out-degree of each node are equal.
	More specifically, the requirements towards the oracle $\mathcal O^{\text{Euler}}$ in conjunction with the input graph $G$ (the one which we want to solve SSSP on) are as follows.
	
	\begin{definition}
		\label{def:euler_oracle}
		Let $H$ be an Eulerian graph that is a sub-graph of some $H'$, where $H'$ is obtained by adding at most $\tilO(1)$ arbitrarily connected ``virtual'' nodes to $G$.	
		Virtual edges, i.e., edges incident to at least one virtual node, are given in distributed form: an edge between a virtual and a real node is known by the latter, and edges between two virtual nodes are known by every node. Then $\mathcal O^{\text{Euler}}$ outputs an orientation of edges of $H$ such that in- and out-degree of each node are equal.
	\end{definition}
	
	Our goal is to show that we can also efficiently implement a call of $\mathcal O^{\text{Euler}}$ in \hybrid. Our goal is to re-purpose the PRAM-model algorithm by \cite{Atallah1984} for this task. The $\mathsf{PRAM}$ model assumes multiple processors and a shared memory containing the input. Processors can write or read a memory cell in each step subject to some restrictions. Each cell may contain a real value but in the SSSP problem with polynomially bounded, integer weights, $\tilO(1)$ bits per cell suffice. 
	The performance of a $\mathsf{PRAM}$ algorithm is measured in the total number of processing steps $N$ (work) and the marginal number of required parallel steps $T$ (depth) if an arbitrary number of processors can be used. In the exclusive read, exclusive write variant (EREW) each memory cell can be read or written by at most one processor in each step. 
	
	Since the memory is the only means by which processors can exchange information, intuitively, an EREW $\mathsf{PRAM}$ algorithm that uses at most $\tilO(n)$ processors and memory cells can be simulated in  $\tilO(1)$ rounds via the $\NCC$ model, where each node simulates $\tilO(1)$ processors \textit{and} cells of the shared memory and \textit{exclusive} access to each cell is provided via the global network (\NCC model). The simulation of  $\mathsf{PRAM}$ algorithms in \hybrid requires $\tilO(N/n+T)$ rounds if each node has the same amount of cells to simulate. 
	
	In \cite{Feldmann2020} it is shown that the simulation of EREW $\mathsf{PRAM}$ can also be done if the input is the local network $G$, however this adds a term of $\tilO(a)$ to the running time, where $a$ is the arboricity of $G$,\footnote{The arboricity of $G$ is defined as the smallest number of forests required to cover all edges of $G$.} which reflects imbalances in the number of memory cells that nodes must simulate. Note that $a$ can be linear in $n$. Their result carries over to the concurrent read/write (CRCW) variant, which can be simulated with EREW incurring only $\tilO(1)$ slowdown. The result of \cite{Feldmann2020} implies the following lemma.
	
	\begin{lemma}[cf.\ \cite{Feldmann2020}]
		\label{lem:sim-pram}
		A CRCW $\mathsf{PRAM}$ algorithm that solves a graph problem on $G=(V,E)$ with depth $T$ using $n = |V|$ processors can be simulated in $\tilO(a)$ rounds in the \hybrid model, w.h.p., where $a$ equals the arboricity of $G$.
	\end{lemma}
	
	We break the problem of Eulerian orientation down to sub problems with suitable properties so that in a final step we can efficiently simulate the PRAM-model algorithm by \cite{Atallah1984}.
	
	\begin{lemma} 
		\label{lem:sim_euler_oracle}
		A call of the oracle $\,\mathcal O^{\text{Euler}}$ (see Definition \ref{def:euler_oracle}) can be implemented in $\tilO(1)$ rounds in the \hybrid model.
	\end{lemma}
	
	\begin{proof}

		There is a CRCW $\mathsf{PRAM}$ algorithm that solves the $\mathsf{Eulerian}$-$\mathsf{Orientation}$ problem with linear work $\tilO(n\p m)$ ($m := |E|$) and depth $\tilO(1)$ \cite{Atallah1984}. 
		Note that in the $\mathsf{PRAM}$ simulation we can safely ignore the fact that some nodes and edges in $H$ are virtual (see Def.\ \ref{def:euler_oracle} for the definition of $H$), because these add at most $\tilO(1)$ nodes and $\tilO(n)$ edges, which can be evenly distributed and simulated by real nodes.
		Our goal is to reduce the problem into sub-problems on sub-graphs of smaller arboricity using the local network, which admits a solution in $\tilO(1)$ rounds simulating the $\mathsf{PRAM}$ algorithm by \cite{Atallah1984} using the simulation in \hybrid (Lemma \ref{lem:sim-pram}). 
		
		The idea is to reduce the arboricity of $H$ by greedily orienting disjoint cycles in $H$ on small subgraphs in parallel, which will ultimately leave us with a remaining graph of yet unoriented edges with small arboricity. Note that by orienting disjoint cycles consistently in one direction, the remaining graph of unoriented edges retains even node degree, i.e., it still has an $\mathsf{Eulerian}$-$\mathsf{Orientation}$. To identify disjoint cycles efficiently, we first compute a $(\chi,D)$-network decomposition. This is a partition of nodes into clusters of diameter $D$ such that the cluster graph (where two clusters are incident if a pair of nodes from each cluster are incident) has a valid $\chi$-coloring. A network decomposition with 
		$D,\chi \in \bigO(\log n)$ can be computed in $\tilO(1)$ rounds in the local network (\LOCAL), see \cite{Linial1991}.
		
		However, we compute such a network decomposition not on $G$, but on the power graph $G^{2} = (V,E')$, where any two nodes in $V$ at distance at most $2$ hops in $G$ share an edge in $E'$. Note that the round complexity is asymptotically the same as computing the network decomposition of $G$, since $G^{2}$ can be simulated in $2$ rounds in the local network. Let $C$ be a cluster of the network decomposition of $G^{2}$. Since two nodes in different clusters of the same color can not be incident in $G^2$, the distance in the \textit{original network} $G$ between such two nodes of  must bigger than $2$.	
		Let $C'$ be the extended cluster of $C$, defined as $C$ and all nodes within one hop of $C$. Note, due to the aforementioned distance property any two such \textit{extended} clusters of the same color are node-disjoint.
		
		For each color in $[\chi]$ and for each cluster $C$ with that color in parallel, all nodes in the extended cluster $C'$ learn $C'$ in $D \in \bigO(\log n)$ steps via the local network. This knowledge enables all nodes in $C'$ to consistently orient and greedily remove all cycles (w.r.t.\ $G$) in $C'$. Note that edges that are nominally removed will still be used to communicate in $C'$. Further, we can do this for each extended cluster $C'$ of the current color in parallel since they do not overlap. Each extended cluster is now cycle-free, thus all edges in all extended clusters of that color can be covered by a single forest.
		Since each node is in some cluster $C$ of some color, both endpoints of any \textit{edge} that was not removed are within some extended cluster $C'$, thus in the end all edges are covered by some forest. The number of forests $a$ corresponds to the number of colors $\chi \in \bigO(\log n)$. Thus the remaining graph can be oriented with the simulation of \cite{Feldmann2020} to run the $\mathsf{PRAM}$ algorithm by \cite{Atallah1984} in $\tilO(1)$ rounds.
	\end{proof}	
	
	We can now apply our implementation of  the $\mathsf{Minor}$-$\mathsf{Aggregation}$ model and the oracle $\mathcal O^{\text{Euler}}$ in \hybrid to the result by \cite{Rozhon2022}, hence, Theorem \ref{thm:almost_shortest_sssp} is proven by combining Lemma \ref{lem:sssp_minor_aggregation}, Lemma \ref{lem:simulate_minor_agregation} and Lemma \ref{lem:sim_euler_oracle}.

	\section{Shortest Paths with Multiple Sources}
	\label{sec:kssp}
	
	In this section we  give a direct application for the fast SSSP algorithm from the previous section by providing a framework to efficiently schedule multiple algorithms in parallel on a so called skeleton graph. Since our algorithmic solutions scale with the amount of global communication that we permit, we also give solutions for the more general \hybridpar{\infty}{\gamma} model. 
	In particular, we show that \smash{$\tilO\big(\!\sqrt{k/\gamma}\big)$} rounds suffice to compute approximations (given that the precision parameter $\eps$ is constant). The claim for the standard \hybrid model is implied by substituting $\gamma = \log^2 n$. Note that the general result is also tight (up to $\tilO(1)$ factors), as the lower bound of \smash{$\tilOm\big(\!\sqrt{k}\big)$} by \cite{Kuhn2020} can be generalized to \smash{$\tilOm\big(\!\sqrt{k/\gamma}\big)$} (see \cite{Schneider2023}).
	
	
	\begin{theorem}
		\label{thm:k-ssp}
		Let $\eps > 0$. The $k$-SSP problem can be approximated w.h.p.\
		\begin{itemize}
			\item in \hybrid in \smash{$\tilO\big(\!\sqrt{k}\cdot\tfrac{1}{\eps^2}\big)$} rounds for stretch $1 \p \eps$ and $k$ i.i.d. \emph{random} sources,
			\item  in \hybridpar{\infty}{\gamma} in $\tilO\big(\!\sqrt{k/\gamma}\cdot\tfrac{1}{\eps^2}\big)$ rounds for stretch $3 \p \eps$ and $k$ arbitrary sources,
			\item  in \hybridpar{\infty}{\gamma} in $\tilO\big(\tfrac{1}{\eps^2}\big)$ rounds for stretch $1 \p \eps$ and $k \leq \gamma$ arbitrary sources.
		\end{itemize}
	\end{theorem}
	
	\subsection{Parallel Scheduling of Algorithms on a Skeleton Graph}
	\label{ssec:scheduling}
	
	Besides the scheme for SSSP introduced in Section \ref{sec:sssp_logtime}, our main technical component is to show that we can efficiently run multiple instances of independent algorithms on a so called skeleton graph in parallel. 
	A skeleton graph consists of skeleton nodes that are sampled with probability $1/x$ from $V$. Skeleton edges are created between skeleton nodes that have a path of length roughly $\tilO(x)$ in $G$ between them with the weight of such an edge corresponding to the distance of that path. 
	Such a skeleton graph has many useful properties, the main ones are that it preserves distances of $G$ and can be computed efficiently in a distributed sense, where each skeleton node knows its incident skeleton edges in $\tilO(x)$ rounds. The formal definition is given as follows.
	
	\begin{definition}[cf. \cite{Ullman1991,Augustine2020a}]
		\label{def:skeleton_graph}
		A skeleton graph $\calS = (V_\calS, E_\calS)$ of $G$, is obtained by sampling each node of $G$ to $V_\calS$ with prob.\ at least $\frac{1}{x}$. The edges of $\calS$ are $E_\calS \!=\! \{ \{u,v\} \!\mid\! u,v\!\in\!V_\calS, \text{hop}(u,v) \!\leq\! h\}$ (for some appropriate $h\in \tilO(x)$) with weights $w_\calS(u,v) = d_{h,G}(u,v)$ for $\{u,v\} \in E_\calS$.
	\end{definition}
	
	In \cite{Augustine2019} it is shown that $\calS$ gives a good approximation of the topology of the graph, in particular, that the distance between sampled nodes in the resulting skeleton graph equals the actual distance in the local graph w.h.p.
	
	\begin{lemma}[cf.\ \cite{Augustine2019}]
		\label{lem:skeleton-graph}
		A skeleton graph $\calS = (V_\calS, E_\calS)$ as given in Definition \ref{def:skeleton_graph} can be constructed in $h\in \tilO(x)$ rounds in the \LOCAL (and thus \HYBRID) model w.h.p. for an appropriately chosen $h\in \tilT(x)$ with the following properties. (a) For any $u,v \!\in\! V$ with $hop(u,v) \!\geq\! h$, there is a shortest path $P$ from $u$ to $v$, s.t. any sub-path $Q$ of $P$ with at least $h$ nodes contains a node in $V_\calS$ w.h.p. (b) For all $u,v, \in V_\calS:$ $d_\calS(u,v) = d_G(u,v)$ w.h.p.
	\end{lemma}
	
	Furthermore, for efficient scheduling of algorithms in parallel we also require the concept of so called \textit{helper sets} that have been introduced in \cite{Kuhn2020}. 
	The rough idea is, when given a set of nodes that have been sampled i.i.d., one can assign each sampled node a set of helper nodes of a certain size within some small distance, which essentially marks the spot where distance and neighborhood size (with respect to some ``load'' $x$) are in a balance. We parameterize the definition of helper sets from \cite{Kuhn2020} as follows.
	
	\begin{definition}[cf.\ \cite{Kuhn2020}]
		\label{def:helpers}
		Let $G = (V,E)$ be a graph, let $x < n$ and assume $W \subseteq V$ was sampled i.i.d.\ from $V$ with prob.\ $1/ x$. 
		A family $\{H_w \subseteq V \mid w \in W\}$ of \emph{helper sets} fulfills the following properties for all $w \in W$ and some integer $\mu \in \smash{\tilT (x)}$.
		(1) Each $H_w$ has size at least $\mu$. (2) For all $u \in H_w$ it holds that $hop(w,u) \leq \mu$. (3) Each node is member of at most $\tilO(1)$ sets $H_w$.
	\end{definition}
	
	The computation of helper sets by \cite{Kuhn2020} works out of the box and we are going to utilize it in the following form.
	
	\begin{lemma}[see \cite{Kuhn2020}]
		\label{lem:helpers}
		A family of helper sets $\{H_w \subseteq V \mid w \in W \}$ as given in Definition \ref{def:helpers} can be computed w.h.p.\ in $\tilO\big(x\big)$  rounds.
	\end{lemma}
	
	We are now ready to prove the following theorem.
	
	\begin{theorem}
		\label{thm:scheduling}
		Let $\gamma < k < n$ and let $\calS$ be a skeleton graph of the local graph $G$ with sampling probability $\sqrt{\gamma/k}$ (see Definition \ref{def:skeleton_graph}). Let $\calA_1, \dots, \calA_k$ be \hybrid algorithms operating on $\calS$ with round complexity at most $T$. Then $\calS$ can be constructed and  $\calA_1, \dots, \calA_k$ can be executed on $\calS$ in $\tilO\big(\sqrt{k/\gamma}\cdot T\big)$ rounds of the \hybridpar{\infty}{\gamma} model w.h.p.
	\end{theorem}
	
	\begin{proof}
		The first step is to compute the skeleton graph $S =(V_\calS, E_\calS)$ as defined in Definition \ref{def:skeleton_graph} with sampling probability \smash{$\sqrt{\gamma/ k}$}, which takes \smash{$\tilO\big(\!\sqrt{k/\gamma}\big)$}  rounds and guarantees $d_\calS(u,v) = d_G(u,v)$ for all $u.v \in V_\calS$ w.h.p., by Lemma \ref{lem:skeleton-graph}.	
		As second step, we compute helper sets for $V_\calS$, which assign each $u \in V_\calS$ a set $H_u \subseteq V$ such that (1) \smash{$|H_u| \in \tilT\big(\!\sqrt{k/\gamma}\big)$} and (2) \smash{$\max_{v \in H_u}\! \hop(u,v) \in \tilT\big(\!\sqrt{k/\gamma}\big)$} and (3) each $v \in V$ is member of $\tilO(1)$ helper sets. This is accomplished w.h.p.\ in \smash{$\tilO\big(\!\sqrt{k/\gamma}\big)$} rounds by Lemma \ref{lem:helpers}.
		
		Note that $\calS$ is established in a distributed sense, where each skeleton edge is known by its incident skeleton nodes. The third step is to transmit the incident edges and the input w.r.t.\ each algorithm $\calA_i$ of $u \in \calS$ to its helpers $H_u$ using the local network, which takes \smash{$\tilO\big(\!\sqrt{k/\gamma}\big)$} rounds due to (2).	This gives each helper $v \in H_u$ all the required information to simulate any $\calA_i$ on $\calS$ on $u$'s behalf, if also provided with the messages that $u$ receives during the execution of $\calA_i$.
		
		In the fourth step, the simulation of the algorithms $\calA_i$ for node $u\in \calS$ is distributed among all helpers $H_u$ in a balanced fashion, that is, each helper $v\in H_u$ is assigned at most $\ell \leq \lceil k/ |H_u| \rceil \in \tilO\big(\!\sqrt{k\gamma}\big)$ algorithms $\calA_i$ to simulate (where $\ell$ is uniform for all helper sets). In particular, we enumerate helpers $H_u$ with indices $v_j$ and assign algorithms $\calA_{\ell (j-1) +1}, \dots, \calA_{\ell j}$ to $v_j$.	
		Furthermore, for each edge $\{u,u'\} \in E_\calS$ and each $\calA_i$, we pair up the helper $v_j \in H_u$ that simulates $\calA_i$ for $u$ with the helper $v_j' \in H_{u'}$ that simulates $\calA_i$ for $u'$.
		By Definitions $\ref{def:skeleton_graph}$ and \ref{def:helpers} (regarding the distances between helpers and between incident skeleton nodes) and the triangle inequality it holds that $$\hop(v_j,v_j') \leq \hop(v_j,u) + \hop(u,u') + \hop(u',v_j') \in \smash{\tilO\big(\!\sqrt {k/\gamma}\big)}$$ and therefore step four can be coordinated in \smash{$\tilO\big(\!\sqrt {k/\gamma}\big)$} rounds via the local network.	
		It remains to coordinate the message exchange among helpers that simulate some $\calA_i$ for some skeleton node $u\in \calS$. In the first round, all outgoing messages can be computed locally by helpers $v_j \in H_u$ based on the input they received from $u$. In general, we presume that the simulation and computation of outgoing messages was correct up to the current round. It suffices to show the correct transmission of all messages for all simulated algorithms for the next round. 
		
		\textit{Messages via the local network $\calS$:} 	
		The hop-distance between helpers $v_j,v_j'$ that simulate $\calA_i$ for two skeleton nodes that are adjacent in $\calS$ nodes is \smash{$\hop(v_j,v'_j) \in \tilO\big(\! \sqrt {k/\gamma}\big)$}. Hence, the local messages can be transmitted in the same number of rounds using the unlimited local network. \textit{Messages via the global network:} each helper is assigned at most \smash{$\ell \in \tilO\big(\!\sqrt{k\gamma}\big)$} algorithms to simulate. Since each helper has a global capacity of $\gamma$ all global messages can be sent sequentially in \smash{$\tilO\big(\!\sqrt{k\gamma} / \gamma\big) = \tilO\big(\!\sqrt{k/\gamma} \big)$} rounds. Due to the canonical assignment of algorithms $\calA_i$ to helpers $v_j$, no helper will receive more messages than it sends in each round.
		
		Consequently, each helper can simulate a single round of all assigned $\calA_i$ in \smash{$\tilO\big(\!\sqrt{k/\gamma} \big)$}  rounds, thus it takes \smash{$\tilO\big(\!\sqrt{k/\gamma} \cdot  T\big)$} rounds to complete the simulation of all $\calA_i$. Finally, any output that was computed by some helper $v_j \in H_u$ in the simulation of $\calA_i$ can be transmitted back to $u$ in \smash{$\tilO\big(\!\sqrt{k/\gamma} \big)$} rounds via the local network.
	\end{proof}
	
	\subsection{Solving the $k$-SSP Problem}
	
	We can now employ the procedure from Section \ref{ssec:scheduling} to schedule $k$ instances of the SSSP algorithm from Section \ref{sec:sssp_logtime} on an appropriately sized skeleton graph $\calS$, which allows us to solve $k$-SSP for the case that the sources are part of the skeleton $\calS$.
	
	\begin{lemma}
		\label{lem:k-SSP_on_skeleton}
		Let $\gamma < k < n$ and let $\calS = (V_\calS, E_\calS)$ be a skeleton graph of the local graph $G$ with sampling probability \smash{$\sqrt{\gamma/k}$} (see Definition \ref{def:skeleton_graph}). Provided that all sources are in $V_\calS$ we can solve the $k$-SSP problem with stretch $1 \p \eps$ in $\tilO\big(\!\sqrt{k/\gamma}\cdot \tfrac{1}{\eps^2}\big)$ rounds of the \hybridpar{\infty}{\gamma} model w.h.p.	
	\end{lemma}
	
	\begin{proof}
		Let $K \subseteq V_\calS$ be the set of $k$ source nodes. Let algorithms $\calA_s$ be an instance of the SSSP algorithm from Section \ref{sec:sssp_logtime} for some source $s \in K$. By Theorem \ref{thm:almost_shortest_sssp} the algorithms $\calA_s$ have a native round complexity of $\tilO(\tfrac{1}{\eps^2})$ to compute an approximation with stretch $1 \p \eps$. We apply Theorem \ref{thm:scheduling} to schedule all $\calA_s, s \in K$ on $\calS$ in parallel in \smash{$\tilO\big(\!\sqrt{k/\gamma} \cdot \frac{1}{\eps^2} \big)$} rounds.	
		After conducting this procedure, every skeleton node $u \in V_\calS$ knows a stretch $1 \p \eps$ distance approximation $\tild_\calS(u,s)$  for every $s \in K$. 
		
		In a post-processing step, each node $v \in V$ that is not part of the skeleton $\calS$ computes its own distance to each source as follows. For this we utilize that edges $E_\calS$ correspond to paths of at most \smash{$h \in \tilO\big(\!\sqrt{k/\gamma} \big)$} hops in $G$.
		First, node $v$ retrieves its distance information to each source from all skeleton nodes within $h$ hops, i.e., in \smash{$\tilO(\sqrt{k/\gamma})$} rounds. For each source $s \in K$, node $v$ computes $\tild(v,s) := \min_{u \in \calS} d_{h,G}(v,u) + \tild_\calS(u,s)$. By Lemma \ref{lem:skeleton-graph} (a) there is a shortest path from $v$ to $s$ which has a skeleton node $u \in V_\calS$ within $h$ hops of $v$ w.h.p. For this particular skeleton node $u$ we have 
		\begin{align*}
			\tild(v,s) & \hspace{4.1mm}\leq \hspace{3.9mm} d_{h,G}(v,u) + \tild_\calS(u,s) \\
			& \hspace*{1.6mm} \stackrel{Thm.\ref{thm:almost_shortest_sssp}}{\leq} \hspace*{1.6mm} d_{h,G}(v,u) + (1 \p \eps)d_\calS(u,s)\\
			&\!\stackrel{Lem.\ref{lem:skeleton-graph}(b)}{=} d_{h,G}(v,u) + (1 \p \eps)d_G(u,s) \\
			& \hspace{4.1mm}\leq \hspace{3.9mm} (1 \p \eps)\big(d_{h,G}(v,u) + d_G(u,s)\big)\phantom{\stackrel{Lem.\ref{lem:skeleton-graph}(b)}{=}}\\ 
			& \!\stackrel{Lem.\ref{lem:skeleton-graph} (a)}{\leq} (1 \p \eps)d_{G}(v,s).\tag*{\qedhere}
		\end{align*}
	\end{proof}

	Finally, we show that solving $k$-SSP on the skeleton $\calS$ can also be used to obtain good approximations for more general cases of the $k$-SSP problem on the local graph $G$ with the different trade offs as stated in Theorem \ref{thm:k-ssp}.
	
	\begin{proof}[Proof of Theorem \ref{thm:k-ssp}]
		The case $k \leq \gamma$ is the simplest. Here, we have enough available global capacity in $\hybridpar{\infty}{\gamma}$ to execute the $k$ \hybrid SSSP algorithms in parallel in $\tilO(\tfrac{1}{\eps^2})$ rounds (Theorem \ref{thm:almost_shortest_sssp}), since each such algorithm requires only $\tilO(1)$ global capacity per round and the local capacity is unlimited.
		
		In the remaining proof we assume $\gamma < k$.		
		Let $\calS$ be a skeleton graph with skeleton nodes $V_\calS$ sampled with probability $\sqrt{\gamma/k}$ and edges $E_\calS$ that correspond to paths of at most \smash{$h \in \tilO\big(\!\sqrt{k/\gamma} \big)$} hops in $G$, which can be computed in  \smash{$\tilO\big(\!\sqrt{k/\gamma} \big)$} rounds (see Definition \ref{def:skeleton_graph} and Lemma \ref{lem:skeleton-graph}). 
		Let $K \subseteq V_\calS$ be the set of $k$ source nodes.
		
		Consider the case where $K$ was sampled i.i.d.\ from $V$ and where we are interested in the standard \hybrid model, i.e., $\gamma \in \tilO(1)$.	
		We make a distinction on the size of $k$. For a large number of sources $k \geq n^{2/3}$ we apply the known $k$-SSP algorithm by \cite{CensorHillel2021a} whose round complexity is \smash{$\tilO(n^{1/3}+\sqrt{k})$}, which equals \smash{$\tilO(\!\sqrt{k})$} for $k$ above this threshold and even provides an exact solution.
		
		If the number of sources is small $k \leq n^{2/3}$ we add $K$ to $V_\calS$ and conduct the construction of the skeleton edges with the set $K$ in addition to the sampled nodes.
		Since each node has a-priori the same chance to be drafted to the skeleton either by being sampled or as a random source, the property of an i.i.d.\ selection of nodes is retained. Note that \smash{$|V_\calS| \in \tilT\big(n/ \sqrt{k}\big)$} w.h.p.\ before adding the sources $K$ to $V_\calS$ (due to the sampling probability \smash{$\tilO\big(1/\sqrt k\big)$} and w.h.p.\ by applying a Chernoff bound). The number of sources added to $V_\calS$  can be bounded by \smash{$k \leq n^{2/3} = n / n^{1/3}\leq n / \sqrt{k}$}. Thus adding $K$ to $V_\calS$ only changes the size of $|V_\calS|$ by a $\tilO(1)$ factor. Therefore, we can apply Lemma \ref{lem:k-SSP_on_skeleton} to obtain a $1 \p \eps$ approximation of $k$-SSP on $G$ with $k \leq n^{2/3}$ random sources.
		
		Finally, consider the case where the $k$ sources are chosen arbitrarily. We have to deal with the fact that these might be highly concentrated in a local neighborhood and can therefore not simply be added to the skeleton as this leads to difficulties with scheduling the SSSP algorithms (even assuming $k \leq n^{2/3}$ does not help here). We employ technique similar to \cite{Kuhn2022}, where each source $s \in K$ tags its closest skeleton node $u_s :=  \argmin_{w \in V_\calS} d_{h,G}(s,w)$ as its \textit{proxy source} in $\calS$ (note that $u_s$ is within \smash{$h \in \tilO\big(\!\sqrt{k}\big)$} hops of $s$ by Lemma \ref{lem:skeleton-graph} (a)).
		
		Subsequently, we have $k' \leq k$ \textit{skeleton} nodes that are tagged as proxy sources, which again allows us to compute $1 \p \eps$ approximations for $k'$-SSP on the proxy sources in $\tilO\big(\!\sqrt{k/\gamma}\cdot \tfrac{1}{\eps^2}\big)$ rounds using Lemma \ref{lem:k-SSP_on_skeleton}. To construct decent approximations to the real source nodes we have to first make the distance $d_{h,G}(u_s,s)$ from each source $s\in K$ to its closest skeleton $u_s$ public knowledge. 
		
		For this we use the so called token dissemination routine (see \cite{Augustine2020a}), which solves the following problem. Nodes have tokens of size $\bigO(\log n)$ (in our case distance labels) which need to be broadcast to everyone. Given that there are at most $x$ such tokens in the network and no node has more than one token, \cite{Augustine2020a} solves this in  $\tilO\big(\!\sqrt{x}\big)$ rounds in the \hybrid model. In the \hybridpar{\infty}{\gamma} we can exploit the higher global capacity by running $y \in \tilO(\gamma)$ such algorithms in parallel, where each algorithm is responsible for broadcasting at most $x \in \tilO(k/\gamma)$ tokens, thus the round complexity is $\tilO\big(\!\sqrt{k/\gamma}\big)$ (note that each token can be assigned i.i.d., randomly to one of the $y$ parallel token dissemination algorithms ensuring $x \in \tilO(k/\gamma)$ w.h.p.).
		
		Now each node $v \in V$ constructs its (approximate) distance to each source $s \in S$ as follows. If there is a shortest path from $v$ to $s$ with at most $h$ hops, then the local network suffices to find it in \smash{$h \in \tilO\big(\!\sqrt{k/\gamma}\big)$} rounds.	
		Else, $v$ computes \smash{$\hat d(v,s) := \tild(v,u_s) + d_{h,G}(u_s,s)$}. Note that there is a node $w \in V_\calS$ within $h$ hops of $s$ on some shortest path from $v$ to $s$. Due to the property of $u_s$ we have $d_{h,G}(u_s,s) \leq d_{h,G}(w,s) \leq d_{h,G}(v,s)$. We combine this with the triangle inequality and obtain
		\begin{align*}
			\hat d(v,s) & \hspace{4.1mm}= \hspace{3.9mm} \tild(v,u_s) + d_{h,G}(u_s,s)\\
			& \hspace*{1.6mm} \stackrel{\text{Thm. }\ref{thm:almost_shortest_sssp}}{\leq} \hspace*{1.6mm} (1 \p \eps)d_{G}(v,u_s) + d_{h,G}(u_s,s)\\
			& \hspace{4.1mm}\leq \hspace{3.9mm} (1 \p \eps)\big(d_{G}(v,w) + d_{G}(w,u_s)+ d_{h,G}(u_s,s)\big)\phantom{\stackrel{Lem.\ref{lem:skeleton-graph}(b)}{=}}\\ 
			& \hspace{4.1mm}\leq \hspace{3.9mm} (1 \p \eps)\big(d_{G}(v,w) + d_{G}(w,u_s)+ d_{h,G}(w,s)\big)\phantom{\stackrel{Lem.\ref{lem:skeleton-graph}(b)}{=}}\\ 
			& \hspace{4.1mm}= \hspace{3.9mm} (1 \p \eps)\big(d_{G}(v,s) + d_{G}(w,u_s)\big)\phantom{\stackrel{Lem.\ref{lem:skeleton-graph}(b)}{=}}\\ 
			& \hspace{4.1mm} \leq \hspace{3.9mm} (1 \p \eps)\big(d_{G}(v,s) + d_{G}(w,s) + d_{G}(s,u_s)\big) \phantom{\stackrel{Lem.\ref{lem:skeleton-graph}(b)}{=}}\\ 
			& \hspace{4.1mm} \leq \hspace{3.9mm} (1 \p \eps)\big(3d_{G}(v,s)\big) \stackrel{\eps' := 3 \eps}{=} (3 \p \eps')d_{G}(v,s). \phantom{\stackrel{Lem.\ref{lem:skeleton-graph}(b)}{=}} \tag*{\qedhere}
		\end{align*}
	\end{proof}

	\bibliographystyle{plain}
	\bibliography{ref/ref}

\end{document}